\documentclass{article}
\usepackage{graphicx}
\usepackage{cite}
\usepackage{amsmath}
\usepackage{amsthm}
\usepackage{amssymb}

\newcommand{\Nat}{{\mathbb N}}
\newcommand{\Real}{{\mathbb R}}
\newcommand{\st}{\;|\;}
\newcommand{\set}[1][ ]{\{ #1 \}}

\newtheorem{theorem}{Theorem}
\newtheorem{example}{Example}

\begin{document}

\title{On the Mediation of Program Allocation in High-Demand Environments}

\author{Fabiano de S. Oliveira$^1$\thanks{Corresponding author (fabiano.oliveira@ime.uerj.br).}\\
Valmir C. Barbosa$^2$\\
\\
$^1$Instituto de Matem\'atica e Estat\'\i stica\\
Universidade do Estado do Rio de Janeiro\\
Rua S\~ao Francisco Xavier, 524, sala 6019B\\
20550-900 Rio de Janeiro - RJ, Brazil\\
\\
$^2$Programa de Engenharia de Sistemas e Computa\c c\~ao, COPPE\\
Universidade Federal do Rio de Janeiro\\
Caixa Postal 68511\\
21941-972 Rio de Janeiro - RJ, Brazil}

\date{}

\maketitle

\begin{abstract}
In this paper we challenge the widely accepted premise that, in order to carry out a distributed computation, say on the cloud, users have to inform, along with all the inputs that the algorithm in use requires, the number of processors to be used. We discuss the complicated nature of deciding the value of such parameter, should it be chosen optimally, and propose the alternative scenario in which this choice is passed on to the server side for automatic determination.  We show that the allocation problem arising from this alternative is NP-hard only weakly, being therefore solvable in pseudo-polynomial time. In our proposal, one key component on which the automatic determination of the number of processors is based is the cost model. The one we use, which is being increasingly adopted in the wake of the cloud-computing movement, posits that each single execution of a program is to be subject to current circumstances on both user and server side, and as such be priced independently of all others. Running through our proposal is thus a critique of the established common sense that sizing a set of processors to handle a submission to some provider is entirely up to the user. 

\end{abstract}

\newpage
\section{Introduction}
\label{intro}

    For about one decade now, the proliferation of multicore architectures has transformed the problem of job scheduling from that of manipulating a single queue with priorities (in order to let processes alternate in using the single available processor) into problems of considerably more elaborate formulations~\cite{Zhuravlev12}. Not only has the number of processing units on a single machine increased, but the number of machines made available as computing clusters has undergone a major increase as well. The adoption of cloud computing worldwide in recent years is leading data centers to become ever larger in order to meet the processing demands of a growing number of companies, sharing resources with them as well as costs. A direct consequence of this is that the processing power now available for use by any given program was unthinkable only a few decades ago, something that no individual organization could have access to exclusively. When used in some carefully limited way, say for a small number of time slots, tapping such processing power is now routinely affordable.
    
    The cost of each execution on such a data center typically depends on the number of processors used and on the total running time. A clear trade-off then arises: while on the one hand it is generally preferable to use more processors, since this tends to lead to shorter running times, on the other hand a greater cost is incurred as the number of processors in increased. Thus, deciding whether increasing the number of processors is worth depends on the running time of the algorithm in question given an input, on how the infrastructure costs vary with the number of processors and running time, and on how much users are willing to spend on their executions. 
    
    We consider a computational environment to which new jobs are submitted for execution from time to time and where they are served by processing power drawn from an infrastructure that is common to all jobs. By ``job'' we mean an algorithm to be run on a given input. In cloud infrastructures, where costs are typically associated with the running times of jobs, such costs should never be affected by the concurrent execution of other jobs, lest the cost model itself fall apart. Here we assume that this noninterference is achieved by a resource-allocation policy that guarantees exclusive access to all resources assigned to any given job. In particular, once a processor is committed to a job, it remains dedicated to that job until the job's completion. We also assume that, with the exception of the processors themselves, all storage-related resources are sufficiently abundant to meet the needs of all jobs. It then follows that, if in a high-demand scenario certain jobs cannot initiate their executions immediately upon being submitted, then this is due exclusively to a temporary lack of available processors. For what follows in this paper, it is immaterial whether the total capacity of such an infrastructure (its number of processors) is time-invariant or -varying. The problem of interest is the following: given a set of jobs to be run, how can the number of processors to be granted each job be decided by an automated procedure on the server side? Let us first formalize this question properly. 
    
    Let $\mathcal{J}$ be a set of jobs concurrently submitted for execution on the processing infrastructure. A \emph{job} $(A, I) \in \mathcal{J}$ consists of an algorithm $A$ and a particular input $I$ for $A$ to run on. For $J=(A, I) \in \mathcal{J}$, let $T_J(N)$ be a function $T_J: \Nat^* \rightarrow \Real^*_+$ giving the overall running time of $A$ on $I$ when the number of processors for the exclusive use of $A$ is $N$ (we therefore assume that such running time is independent of the particular time at which the submission of $J$ takes place). It is worth mentioning that, for $N = 1$,  the asymptotic value of $T_{(A, I)}(N)$ with respect to the size of $I$, once stripped of any constant factor, is the running time used in standard computational-complexity analyses of $A$. Likewise, letting $N$ vary forms the basis of the standard speedup analyses of parallel algorithms. 
    
    The function $T_{(A, I)}(N)$ will play such an important role in the process of automating the selection of the number of processors for each job that we pause momentarily to comment on the issue of how feasible it is to obtain $T_{(A, I)}(N)$ in the first place. This is a legitimate concern that can arise in many cases, since predicting the running-time function of an algorithm given an input to it is at present essentially a research problem. What prompts us to proceed in spite of such incipience is the fact that there have already been significant advances in the field (cf., e.g.,~\cite{Goldsmith07, Hutter14}, where the authors survey and present their own technique). We feel like such advances do indeed justify our assumption that an algorithm's running-time function, or a reasonable approximation thereof, is known a priori.
    
	Still with regard to the running-time function $T_{(A, I)}(N)$, we introduce the notion of the \emph{saturation point} of job $J=(A,I)$. This point is the value of $N$, the number of processors, beyond which the running time of algorithm $A$ on input $I$ becomes dominated by the communication and/or coordination overhead among the processors, so that increasing $N$ no longer leads to a decrease in $T_{(A, I)}(N)$. In other words, letting $S_J$ denote the saturation point of $J$, we have $T_J(N) > T_J(N+1)$ for $1 \leq N < S_J$ and $T_J(N) \leq T_J(N+1)$ for $N \geq S_J$. We set $S_J = \infty$ whenever $J$ remains unaffected by the aforementioned overhead indefinitely as $N$ grows.
	
    The next major player in our modeling effort is inspired by the recent surge of interest in that form of computing that has become known by the charming name of cloud computing. Regardless of whether any significantly new technology lies behind the name, it remains an undisputed fact that the cloud computing movement has given rise to a new way to approach the handling of an installation's processing power, viz., the pay-as-you-go concept. In this approach, users are charged according to a cost associated with each of their jobs, which in turn depends basically on the actual number of processors involved in executing the job in question and on how long they require to complete it (the job's total running time). 
    
    This approach to cost assessment is reflected in our model via $\mathrm{CoI}(T, N)$, a function that yields the cost of the infrastructure for a job given the time $T$ it takes to run on $N$ processors. That is, $\mathrm{CoI}: \Real^*_+ \times \Nat^* \rightarrow \Real_+$. This function is specified by the infrastructure provider and may vary with time in order to reflect, for instance, the evolving relationship between the demand and the cost of maintaining the physical infrastructure. The $\mathrm{CoI}(T, N)$ function is therefore the provider's ``pricing table,'' being in principle independent of any particular job. We assume that $\mathrm{CoI}(T, N)$ is any function for which the following, called the \emph{domination property}, holds:
\begin{eqnarray*}
	 \text{if } T_1 \leq T_2 \text{ and } N_1 \leq N_2\text{, then }  \mathrm{CoI}(T_1, N_1) \leq \mathrm{CoI}(T_2, N_2).
\end{eqnarray*}

	With the definitions of both $T_J(N)$ and $\mathrm{CoI}(T, N)$ in place, we express the cost of job $J \in \mathcal{J}$ in terms of evaluating $\mathrm{CoI}(T, N)$ on the set of points $\set [(T_J(N), N) \st N \in \Nat^*]$. That is, the cost of executing $J$ on any number of processors $N$ is given by the function $C_J: \Nat^* \rightarrow \Real_+$ such that 
\begin{eqnarray*}
	 C_J(N) = \mathrm{CoI}(T_J(N), N) \text{.}
\end{eqnarray*}

	We proceed, in Section~\ref{deciding}, to a discussion of the mathematics behind the optimal choice of the number of processors for a given job. Such discussion makes the difficulty involved in the determination of the optimal $N$ evident, leading one to question the status quo, which is to relegate to the user the task of coming up with the desired value. We propose instead that a user should only be required to express, in quantitative terms, some measure of how much the job in question is worth. Based on this, the optimal number of processors could be automatically optimized for the benefit of both users and infrastructure providers. We discuss the computational hardness of such optimization in Section~\ref{hardness} and conclude in Section~\ref{conclusion}.

    As a final introductory note, we find it worth mentioning that the allocation problem we consider is fundamentally different from the well-known and prolifically studied scheduling problem, on which comprehensive surveys date back to the 1970s (cf., e.g.,~\cite{Graham79, Graham78}). In the scheduling problem, a set of jobs have to be executed and the goal is either to minimize the time to complete all jobs or, by assigning bonuses and penalties to each job's starting/finishing times, to maximize the bonuses minus the penalties. Many efficient algorithms as well as hardness results are known for particular subproblems~\cite{Rodrigues14}, all invariably avoiding the difficulties associated with a variable number of processors given a job and the ensuing variability in running times. A few recent works have considered some degree of variability in the problem's goals and a job's completion time (cf., e.g.,~\cite{Kuo08, Nouri12, Choi13, Ma16}), but the latter has had to do with problem-specific issues only, not with the allocation of the job to multiple processors that we consider.
	
\section{The problem of deciding the number of processors}
\label{deciding}

	Let us initially consider the hypothetical scenario in which the cost $C_J(N)$ of a job, defined in Section~\ref{intro}, is not taken into account when deciding the value of $N$ that is most appropriate. In such a scenario the infrastructure is assumed to exist for the users' sole benefit without any regard to any economic aspect related to executing jobs on it. Such jobs incur no cost whatsoever ($C_J(N)=0$ for any $J$ or $N$), regardless of how many processors they use or for how long. This being the case, the objective function for users to minimize is quite clear: simply request as many processors as needed to minimize each job's running time, even if overall the jobs end up demanding more processing capacity from the infrastructure than can possibly be provided. Not only is this problematic per se, but the situation tends to get more and more aggravated by the fact that using increasingly more processors for job $J$, even if below the job's saturation point $S_J$, tends to be progressively less rewarding.

	In the realistic scenario of $C_J(N)>0$ for every $J$ and every $N$, a user's goal becomes, if anything, even more fraught with difficulties. If, on the one hand, it is desirable for the user to minimize the cost of executing a job, possibly leading to unfeasibly long running times, on the other hand aiming to minimize the job's running time may lead to unfeasibly high costs. Reasonable strategies will therefore admit higher-than-optimal running times in order to strike a balance in which both the cost of a job and its running time are feasible. Informally, the goal of such a strategy can be expressed in economical terms as follows: given a job, a metric of how much running it to completion is worth from the user's perspective, and the current cost of the infrastructure as given by the $\mathrm{CoI}$ function of Section~\ref{intro}, how many processors are currently optimal for the job to run on?
	
    Let us formalize this question, which asks for a decision on the number of processors to be used to run a job while minding the trade-off between the job's cost and its running time. First of all, clearly using a number of processors that surpasses the job's saturation point is never worth. For a number of processors equal to or below the saturation point, we assume that the user provides, along with each particular job $J$ to be executed, a \emph{utility function} $U_J: \Real_+ \rightarrow \Real_+$ that reflects the trade-off in the following way. \emph{For $a, b \in \Real_+$ such that $a < b \leq T_J(S_J)$, running time $a$ is preferable to running time $b$ if and only if the additional cost incurred is at most $\int_a^b U_J(t)\;\mathrm{d}t$}.
	
	In other words, in order to run job $J \in \mathcal{J}$, and for $N\in\set[ 2, \ldots, S_J ]$, using $N$ processors is preferable to using $N-1$ processors if and only if
\begin{eqnarray*}
	 C_J(N) - C_J(N-1) \leq \int_{T_J(N)}^{T_J(N-1)} U_J(t)\;\mathrm{d}t{.}
\end{eqnarray*}
It readily follows from this that, more generally and for $1 \leq N_1 < N_2 \leq S_J$, running $J$ on $N_2$ processors is preferable to running it on $N_1$ processors if and only if
\begin{eqnarray*}
	  C_J(N_2) - C_J(N_1) \leq \int_{T_J(N_2)}^{T_J(N_1)} U_J(t)\;\mathrm{d}t{.}
\end{eqnarray*}
The role of the utility function is therefore to make explicit how valuable a faster execution is to the user. And while both $T_J(N)$ and $C_J(N)$ can be expected to vary only negligibly among jobs that share the same underlying algorithm $A$ and input $I$, the utility function $U_J(t)$ is inherently dependent on current circumstances and, as such, is capable of reflecting the user's predisposition to wait for the job's output.    
	
	Once a measure of the value attributed to jobs by their users is available, the process of deciding the number of processors on which to execute each new job is, in principle, amenable to being automated. To see this, let $N_J(N_\mathrm{max})$ be the optimal number of processors for an execution of job $J$ when $N_\mathrm{max}$ is the current maximum number processors available for job $J$, as informed by the infrastructure provider. Clearly, $N_J(N_\mathrm{max})$ can be computed as
\begin{eqnarray*}
	N_J(N_\mathrm{max}) &=& \left \{ 
    \begin{array}{ll} 
	0 & \text{, if } N_\mathrm{max} = 0; \\
	N_\mathrm{max} & \text{, if } C_J(N_\mathrm{max}) - C_J(N^*) \leq \int_{T_J(N_\mathrm{max})}^{T_J(N^*)}  U_J(t)\;\mathrm{d}t; \\
	N^* & \text{, otherwise, }
	\end{array} 
    \right .
\end{eqnarray*}
where $N^* = N_J(N_\mathrm{max}-1)$, that is, $N^*$ is the number of processors that would be the optimal to execute $J$ if the number of available processors fell short of $N_\mathrm{max}$ by $1$. Expressing $N_J(N_\mathrm{max})$ through the above recursion immediately implies that it can be computed by dynamic programming involving $\Theta(N_\mathrm{max})$ calculations of the cost function $C_J$ and of the definite integral of $U_J(t)$.

	We now set out to examine a series of examples that illustrate this use of the function $U_J(t)$ to optimally select the number of processors on which to execute $J$. Continuing in this way requires the following additional notation. We use $\mathcal{N}^*_J$ to denote the set of all possible values of $N_J(N_\mathrm{max})$ as $N_\mathrm{max}$ is varied onward from $0$. That is, $\mathcal{N}^*_J = \set [ N_J(N_\mathrm{max}) \st N_\mathrm{max} \in \Nat ]$.

\begin{example}
\label{Ex1} 
Let $J_1$ be a job whose total workload can always be evenly distributed among any number $N$ of processors. Assume $\mathrm{CoI}(T, N) = kTN$ for some $k \in \Real^*_+$ and that the time to distribute the job's workload to the $N$ processors is negligible.
\end{example}

	In this case we clearly have $T_{J_1}(N) = T_{J_1}(1) / N$, and therefore $C_{J_1}(N)=\mathrm{CoI}(T_{J_1}(N), N)=k N T_{J_1}(1) / N=k T_{J_1}(1)$. That is, the cost of job $J_1$ is invariant with respect to the number of processors used. Consequently, $\mathcal{N}^*_{J_1} = \set [N \in \Nat \st N \leq S_{J_1}]$, with $S_{J_1}=\infty$, regardless of $U_{J_1}(t)$.

	The running-time function used in Example~\ref{Ex1} yields $T_J(1)/T_J(N) = N$, which is the highest speedup that any job $J$ can achieve when executing on $N$ processors. This is fully consistent with the assumption that the total workload of job $J_1$ can always be evenly divided among the processors in use without requiring any additional time (i.e., without any communication or coordination overhead). A generalization of this running-time function yielding a speedup of at most $N^\alpha$ for $0 < \alpha \leq 1$ is the function $T_J(N) = (T_J(1) - T_J(\infty))/N^{\alpha} + T_J(\infty)$, where $T_J(\infty)=\lim_{N \to \infty} T_J(N)$. Note, with regard to this function, that requiring $\alpha \leq 1$ is necessary to prevent speedups beyond $N$ when $T_J(\infty) = 0$. This function is used in the next example, but forbidding the co-occurrence of $\alpha=1$ and $T_J(\infty) = 0$ (so that Example~\ref{Ex2} does not generalize Example~\ref{Ex1}).

\begin{example}
\label{Ex2}
Let $J_2$ be a job such that $U_{J_2}(t) = 0$ for $0 \leq t \leq K$ and $U_{J_2}(t) = \infty$ for $t > K$. Assume $K>T_{J_2}(\infty)$, $\mathrm{CoI}(T, N) = kTN$ for some $k \in \Real^*_+$, and $T_{J_2}(N) = (T_{J_2}(1) - T_{J_2}(\infty))/N^{\alpha} + T_{J_2}(\infty)$ with $\alpha < 1$ or $T_{J_2}(\infty)>0$. 
\end{example}

	Note, in this example, that $S_{J_2}=\infty$, so increasing $N$ eventually leads to $T_{J_2}(N)\leq K$. Let $N_K = \max \set [N \in \Nat \st T_{J_2}(N) > K] + 1$. That is, $N_K$ is the least number of processors for which the running time of $J_2$ is not strictly above $K$. When comparing two numbers of processors $N_1$ and $N_2$ such that $N_1<N_2$, there are two cases to be considered. The first case is that of $N_1<N_K$, which clearly leads to $\int_{T_{J_2}(N_2)}^{T_{J_2}(N_1)} U_{J_2}(t)\;\mathrm{d}t = \infty$ and to $N_2$ being preferable to $N_1$. The second one is that of $N_1 \geq N_K$, which yields $\int_{T_{J_2}(N_2)}^{T_{J_2}(N_1)} U_{J_2}(t)\;\mathrm{d}t = 0$ and therefore $C_{J_2}(N_2) - C_{J_2}(N_1) \leq \int_{T_{J_2}(N_2)}^{T_{J_2}(N_1)} U_{J_2}(t)\;\mathrm{d}t \iff C_{J_2}(N_2) \leq C_{J_2}(N_1)$. But since
\begin{equation*}
	 	C_{J_2}(N) = \mathrm{CoI}(T_{J_2}(N), N) = kN\left(\frac{T_{J_2}(1) - T_{J_2}(\infty)}{N^{\alpha}} + T_{J_2}(\infty) \right) \text{,}
\end{equation*}
\noindent whose first derivative with respect to $N$ is
\begin{equation*}
	 	C'_{J_2}(N) = k\left((1-\alpha)\frac{T_{J_2}(1) - T_{J_2}(\infty)}{N^{\alpha}} + T_{J_2}(\infty)\right) \geq 0, 
\end{equation*}
we have that $C_{J_2}(N)$ is a nondecreasing function. Therefore, $C_{J_2}(N_2) \leq C_{J_2}(N_1)$ $\iff$ $C_{J_2}(N_2) = C_{J_2}(N_1)$ $\iff$ $\alpha = 1$ and $T_{J_2}(\infty) = 0$, which does not hold and therefore $N_2$ is not preferable to $N_1$. Thus, $\mathcal{N}^*_{J_2} = \set [N \in \Nat \st N \leq N_K]$.

\begin{example}
\label{Ex3}
Let $J_3$ be any job and assume $\mathrm{CoI}(T, N) = K \in \Real^*_+$.
\end{example}

	In this case we have $C_{J_3}(N) = K$ regardless of $N$, so $C_{J_3}(N_2) - C_{J_3}(N_1) \leq \int_{T_{J_3}(N_2)}^{T_{J_3}(N_1)} U_{J_3}(t)\;\mathrm{d}t$ holds for all $N_1 < N_2$ with any $U_{J_3}(t)$. Therefore, $\mathcal{N}^*_{J_3} = \set [N \in \Nat \st N \leq S_{J_3}]$.
	
	Examples~\ref{Ex1}--\ref{Ex3} are all straightforward, with results obtained essentially from common sense. Indeed, Example~\ref{Ex1} illustrates no more than the idealized case in which maximum speedups are achievable, always at the same cost, for as many processors as can be obtained from the infrastructure provider. Example~\ref{Ex2}, being far more realistic, illustrates the commonly accepted situation in which reducing a job's running time requires further expenditure. This example also illustrates the possible existence of a running-time ``threshold'' beyond which no amount of improvement can be accepted. Example~\ref{Ex3}, finally, addresses the (again idealized) infrastructures whose cost does not depend on how much of its resources is used. In this scenario, as in that of Example~\ref{Ex1}, acting greedily is always best. Example~\ref{Ex4}, given next, is substantially more complex.
	
\begin{example}
\label{Ex4}
Let $J_4$ be a job such that $U_{J_4}(t) = a\in \Real^*_+$ for $t \geq 0$. Assume $\mathrm{CoI}(T, N) = kTN$ for some $k \in \Real^*_+$ and $T_{J_4}(N) = (T_{J_4}(1) - T_{J_4}(\infty))/N^{\alpha} + T_{J_4}(\infty)$ with $\alpha < 1$.
\end{example}

	In order to characterize $\mathcal{N}^*_{J_4}$, suppose initially that $\overline{N} \in \mathcal{N}^*_{J_4}$ and consider the least natural number $N>\overline{N}$ such that $N \in \mathcal{N}^*_{J_4}$. Therefore,
\begin{eqnarray*}
	 	\lefteqn{C_{J_4}(N) - C_{J_4}(\overline{N}) \leq \int_{T_{J_4}(N)}^{T_{J_4}(\overline{N})} U_{J_4}(t)\;\mathrm{d}t} \\
	 	&\iff& kNT_{J_4}(N) - k\overline{N}T_{J_4}(\overline{N}) \leq a(T_{J_4}(\overline{N}) - T_{J_4}(N)) \\
	 	&\iff& (kN + a)T_{J_4}(N) \leq (k\overline{N} + a)T_{J_4}(\overline{N}) \\
	 	&\iff& (kN + a) ((T_{J_4}(1) - T_{J_4}(\infty)) N^{-\alpha} + T_{J_4}(\infty)) \leq (k\overline{N} + a) T_{J_4}(\overline{N}) \\
	 	&\iff& (kN + a) (T_{J_4}(1) - T_{J_4}(\infty) + T_{J_4}(\infty)N^{\alpha}) \leq (k\overline{N} + a) T_{J_4}(\overline{N})N^{\alpha}.
\end{eqnarray*}
The latter is an inequality on $N^{\alpha}$, $N$, and $N^{1+\alpha}$, involving additionally an independent term. Although specific cases of this inequality can be solved algebraically (e.g., for $\alpha = 0.5$, let $y = N^{0.5}$ to obtain a cubic inequality on $y$), tackling the general case can be complicated. However, we can show that, in general, $\mathcal{N}^*_{J_4} = \set [ N \in \Nat \st N \leq \overline{N}_{J_4} ]$ for some natural number $\overline{N}_{J_4}$.
		
		Recalling that
\begin{eqnarray*}
	 	 \lefteqn{C_{J_4}(N) - C_{J_4}(\overline{N}) \leq \int_{T_{J_4}(N)}^{T_{J_4}(\overline{N})} U_{J_4}(t)\;\mathrm{d}t} \\
         &\iff& (kN + a)T_{J_4}(N) \leq (k\overline{N} + a)T_{J_4}(\overline{N}),
\end{eqnarray*}
and letting $f(N) = (kN + a)T_{J_4}(N)$, we see that the problem becomes characterizing the values of $\overline{N}$ for which $f(N) \leq f(\overline{N})$. We first differentiate $f(N)$ with respect to $N$, obtaining
\begin{eqnarray*}
	 	 f'(N) = k T_{J_4}(N) + (kN + a)T_{J_4}'(N),
\end{eqnarray*}
then note that
\begin{eqnarray*}
	 	 T_{J_4}(N) = \frac{N T_{J_4}'(N)}{-\alpha} + T_{J_4}(\infty),
\end{eqnarray*}
whence
\begin{eqnarray*}
	 	 f'(N) = k T_{J_4}(\infty) - T_{J_4}'(N) \left( kN \left( \frac{1}{\alpha} - 1  \right) - a \right).
\end{eqnarray*}
Because $T_{J_4}'(N) < 0$, we have $f'(N) > 0$ for $kN \left( \frac{1}{\alpha} - 1  \right) - a > 0$ $\iff$ $N > \frac{a}{k}\left( \frac{\alpha}{1 - \alpha}  \right)$. For $N \leq \frac{a}{k}\left( \frac{\alpha}{1 - \alpha}  \right)$, and using

\begin{eqnarray*}
	 	 T_{J_4}'(N) = \frac{N T_{J_4}''(N)}{-(\alpha+1)},
\end{eqnarray*}
we obtain
\begin{eqnarray*}
	 	 f''(N) & = & - T_{J_4}''(N) \left( kN \left( \frac{1}{\alpha} - 1  \right) - a \right) - T_{J_4}'(N) k \left( \frac{1}{\alpha} - 1  \right) \\
	 	 & = & - T_{J_4}''(N) \left( kN \left( \frac{1}{\alpha} - 1  \right) - a \right) + \frac{N T_{J_4}''(N) }{\alpha+1} k \left( \frac{1}{\alpha} - 1  \right).
\end{eqnarray*}
Since $T_{J_4}''(N) > 0$, we have $f''(N) > 0$. Thus, $f'(N)$ has at most one zero in the interval $[0 \;,\; \frac{a}{k}\left( \frac{\alpha}{1 - \alpha}  \right)]$. Let $N_0$ be such that $f'(N_0) = 0$, if such a zero exists, or $N_0 = 1$, if none exists. Let $N' = \lceil N_0 \rceil - 1$. Then $f'(N) < 0$ for all $0 \leq N \leq N'$ and thus $\set[0,\ldots,N'] \subseteq \mathcal{N}^*_{J_4}$. Moreover, since $f'(N) \geq 0$ for all $N \geq N'+1$, then $N \notin \mathcal{N}^*_{J_4}$ if $N > N'+1$. Finally, $N' + 1 \in \mathcal{N}^*_{J_4} \iff f(N'+1) \leq f(N')$. Therefore, for $\overline{N}_{J_4}$ equal to $N'$ or $N'+1$, it holds that $\mathcal{N}^*_{J_4} = \set [ N \in \Nat \st N \leq \overline{N}_{J_4} ]$.

Examples~\ref{Ex1}--\ref{Ex4} have all addressed the issue of computing the set $\mathcal{N}^*_J$ and, essentially, have all resulted in $\mathcal{N}^*_J = \set [N \in \Nat \st N \leq N^+]$ for some $N^+ \leq S_J$. This holds in spite of the fact that $S_J$, in all cases but that of Example~\ref{Ex3}, lies at infinity. In this regard, the case of Examples~\ref{Ex2} and~\ref{Ex4} is particularly curious, because this somewhat degenerate placement of $S_J$ is due to the power-law-decaying running-time function assumed in those examples. On the other hand, the difference between Examples~\ref{Ex2} and~\ref{Ex4} lies in the functional form assumed for $U_J(t)$ in each case (a step from $0$ to $\infty$ at $t=K$ in the former case, a positive real constant in the latter). We finalize the section with a characterization of the ``least'' $U_J(t)$ for which $\mathcal{N}^*_J = \set [N \in \Nat \st N \leq S_J]$, assuming $\mathrm{CoI}(T, N) = kTN$ for some $k \in \Real^*_+$ and $T_{J}(N) = (T_{J}(1) - T_{J}(\infty))/N^{\alpha} + T_{J}(\infty)$ with $0<\alpha\leq 1$.

	To achieve this, first let $\overline{N} \in \mathcal{N}^*_J$. The least $U_J(t)$ for which using $N>\overline{N}$ processors is preferable to using $\overline{N}$ processors is that for which
\begin{eqnarray*}
	  C_J(N) - C_J(\overline{N}) = \int_{t}^{T_J(\overline{N})} U_J(u)\;\mathrm{d}u,
\end{eqnarray*}
where $t = T_J(N)$. Differentiating both sides of this equation with respect to $t$, and taking into account the fact that $\int_{t}^{T_J(\overline{N})} U_J(u)\;\mathrm{d}u = - \int^t_{T_J(\overline{N})} U_J(u)\;\mathrm{d}u$, leads to
\begin{equation*}
	  \frac{\mathrm{d}C_J(N)}{\mathrm{d}t} = - U_J(t),
\end{equation*}
and thence to
\begin{equation*}
	  U_J(t) = - k\left(\frac{\mathrm{d}N}{\mathrm{d}t}t + N\right). 
\end{equation*}
Using $N(t) = \left(\frac{t - T_{J}(\infty)}{T_{J}(1) - T_{J}(\infty)} \right)^{-1/\alpha}$ yields the desired utility function,
\begin{equation*}
	  U_J(t) = k\left(\frac{t/\alpha}{t-T_J(\infty)}-1\right)N(t).
\end{equation*}

	Letting $T_J(\infty)=0$ in this expression allows us to see its significance more clearly, since it leads to $U_J(t)$ being proportional to $N(t)$, with the proportionality constant depending on $k$ and $\alpha$:
\begin{equation*}
	  U_J(t)=k\left(\frac{1}{\alpha}-1\right)N(t).    
\end{equation*}
That is, the utility function $U_J(t)$ that acts as a ``threshold'' between the rejection and the acceptance of a larger optimal number of processors into $\mathcal{N}^*_J$ is, in the case of $T_J(\infty)=0$, proportional to the function that is inverse to $T_J(N)$.
	
\section{The complexity of automating allocation on the server side}
\label{hardness}

	As far as assigning processors to jobs is concerned, the common practice of infrastructure providers has been to have only the somewhat passive role of simply imposing an upper bound on the maximum number of processors that any particular job is allowed to request. Even though such a degree of passivity does not necessarily imply a poor arrangement between provider and users, it does not imply a good one either. In fact, it seems clear that a more active provider could in principle be able to at least pursue an arrangement leading to the best possible benefits for both providers and users. Doing this would require tackling the following question, which arises when the total demand for computing power exceeds the current capacity of the infrastructure provider: given a set of jobs, how is the available computing power going to be parceled out among them so as to both satisfy the users and maximize revenue? This question is formalized as follows.
	
	\vspace{\baselineskip}
	\textbf{Problem:} \textsc{Allocation Problem (AP)}
    
	\textbf{Input:} $K \in \Real$, the number $\mathit{MAXN}$ of available processors, a set $\mathcal{J}$ of jobs, a running-time function $T_J(N)$ for each $J \in \mathcal{J}$, a utility function $U_J(t)$ for each $J \in \mathcal{J}$, and a cost-of-infrastructure function $\mathrm{CoI}(T, N)$.
    
	\textbf{Question:} Is there $N^\mathrm{opt}_J \in \mathcal{N}^*_J$ for every $J \in \mathcal{J}$ so that $\sum_{J \in \mathcal{J}} C_J(N^\mathrm{opt}_J) \geq K$ and $\sum_{J \in \mathcal{J}} N^\mathrm{opt}_J \leq \mathit{MAXN}$?	
	
    \vspace{\baselineskip}
	We show in this section that AP is equivalent to a generalization of the well-known Knapsack Problem (KP) in which the items to be packed are versioned. KP is one of the classical NP-complete problems~\cite{Garey79}. In its optimization version, KP is stated as follows. Given $W \in \Nat$ and a set of $M$ items, the $i$th one having weight $w_i \in \Nat$ and value $v_i \in \Nat$, determine the set $\mathcal{S} \subseteq \set[1,\ldots, M]$ that maximizes $\sum_{i \in \mathcal{S}} v_i$ while ensuring that $\sum_{i \in \mathcal{S}} w_i \leq W$.
	
	In our versioned-item generalization of KP, the $i$th item exists in $V_i$ versions, each having a distinct weight and value. Deciding which version of each item to select for the knapsack (if any) is part of the problem. In formal terms, the following is our generalization of KP.
	
	\vspace{\baselineskip}
	\textbf{Problem:} \textsc{Knapsack with Versioned Items Problem (KVIP)}
	
	\textbf{Input:} $K' \in \Real$, $W \in \Nat$, and a set of $M$ items, the $i$th one having $V_i$ versions, the $j$th of these having weight $w^j_i \in \Nat$ and value $v^j_i \in \Real$.

	\textbf{Question:} Are there $\mathcal{S} \subseteq \set[1,\ldots, M]$ and $f: \set[1,\ldots, M] \rightarrow \Nat$, with $1\leq f(i)\leq V_i$ for every $i\in \mathcal{S}$, such that $\sum_{i \in \mathcal{S}} v^{f(i)}_i \geq K'$ and $\sum_{i \in \mathcal{S}} w^{f(i)}_i \leq W$?
    
	\vspace{\baselineskip}
	KVIP is clearly an NP-complete problem, since KP trivially reduces to KVIP by letting $V_i = 1$ for $1 \leq i \leq M$. Moreover, the classical pseudo-polynomial-time algorithm for the optimization version of KP, based on dynamic programming and requiring $O(WM)$ time, extends naturally to the optimization version of KVIP. Indeed, it suffices to use the recurrence
\begin{eqnarray*}
\lefteqn{\text{KVIP}(W, K', \mathcal{X}) = \text{\textsc{Yes}}} \\
&\iff& \text{KVIP}(W, K', \mathcal{X} \setminus \set[(w_M^{V_M},v_M^{V_M})]) = \text{\textsc{Yes}} \\ 
&& \text{or } \\
&& \text{KVIP}(W-w_M^{V_M}, K'-v_M^{V_M}, \mathcal{X} \setminus \set[(w_M^{j},v_M^{j}) \mid 1 \leq j \leq V_M]) = \text{\textsc{Yes}} \text{,}
\end{eqnarray*}
where $\mathcal{X} = \set[(w_i^{j_i}, v_i^{j_i}) \mid 1 \leq i \leq M, 1 \leq j_i \leq V_i ]$, with base cases
\begin{gather*}
\text{KVIP}(W, K', \emptyset) = \text{\textsc{Yes}} \iff K' = 0,\\
W < 0 \implies \text{KVIP}(W, K', \mathcal{X}) = \text{\textsc{No}},\\
W \geq 0, K' \leq 0 \implies \text{KVIP}(W, K', \mathcal{X}) = \text{\textsc{Yes}}.
\end{gather*}

Thus, the optimization version of KVIP can be solved in pseudo-polynomial time by $O(W \sum_{i=1}^M V_i)$-time dynamic programming. It can also be used to solve the optimization version of AP by means of the following transformation of an AP instance into a KVIP instance:
\begin{itemize}
\item Let $K'=K$, $W=\mathit{MAXN}$ (the total number of processors), and $M=|\mathcal{J}|$ (the number of jobs). For each job $J_i \in \mathcal{J}$, with $1 \leq i \leq M$, the set of optimal numbers of processors $\mathcal{N}^*_{J_i}$ can be computed via dynamic programming based on the recursion for $N_{J_i}(N_\mathrm{max})$, as mentioned previously. Note that $\mathcal{N}^*_{J_i}$ can be assumed finite in real applications, since the number of processors at any infrastructure can never be arbitrarily large.
\item For $1 \leq i \leq M$, let $V_i=|\mathcal{N}^*_{J_i}|$. For $1 \leq i \leq M$ and $1 \leq j \leq V_i$, let $w_i^j$ be the $j$th smallest member of $\mathcal{N}^*_{J_i}$ and $v_i^j = C_{J_i}(w_i^j)$.
\end{itemize}
It then follows that the resulting instance of KVIP leads to a \textsc{Yes} answer if and only if there exists $N^\mathrm{opt}_J \in \mathcal{N}^*_J$ for every $J \in \mathcal{J}$ such that $\sum_{J \in \mathcal{J}} C_J(N^\mathrm{opt}_J) \geq K$ and $\sum_{J \in \mathcal{J}} N^\mathrm{opt}_J \leq \mathit{MAXN}$, that is, if and only if the instance of AP leads to a \textsc{Yes} answer as well. 

Fully characterizing the aforementioned equivalence between AP and KVIP requires, additionally, that every instance of KVIP be similarly transformable into an instance of AP. This is given next, as part of the NP-completeness proof of AP.

\begin{theorem}
\label{APNPC}
	\textsc{AP} is NP-complete.
\end{theorem}
\begin{proof}
	AP is trivially a member of NP. We argue for NP-completeness by displaying a polynomial-time reduction from an instance $I'$ of KVIP to an instance $I$ of AP.
	
	Note initially that we can safely assume, for $1 \leq i \leq M$, that no $r,s$ exist such that $r\neq s$, $1 \leq r, s \leq V_i$, $w^r_i \leq w^s_i$, and $v^r_i \geq v^s_i$. If this were so, then the version $s$ of item $i$ could be discarded, since replacing it by version $r$ in any valid selection would lead to an equally valid and no worse selection. Therefore, for any $i,j$ such that $1 \leq i \leq M$ and $1 \leq j < V_i$, we assume $w^j_i < w^{j+1}_i$ and $v^j_i < v^{j+1}_i$. It will also be convenient to assume that each item has a degenerate version (say, the first one) having no weight or value, that is, $(w_i^1, v_i^1) = (0, 0)$ for all $1 \leq i \leq M$. 

	The reduction to $I$ from $I'$ is as follows (the reader may find it useful to check the reduction example depicted in Figure~\ref{fig:ExemploReducao} while following the proof; the details of instance $I'$ appear in Figure~\ref{fig:ExemploReducao}(a)). Let $K=K'$ and $\mathit{MAXN}=W$. For all $1 \leq i \leq M$, let $J_i$ be a job of $\mathcal{J}$, with saturation point $S_{J_i} = S = \max \set [ w^j_i \st 1 \leq i \leq M, 1 \leq j \leq V_i]$. Let $c^N_i = C_{J_i}(N) = v^j_i$, if $N = w^j_i$ for some $1 \leq j \leq V_i$, or $c^N_i = C_{J_i}(N) = \infty$, otherwise (Figure~\ref{fig:ExemploReducao}(b)). This choice of a job's cost function is supported by appropriate running-time, cost-of-infrastructure, and utility functions, all given next.
	
	Let $T: \Nat^* \rightarrow \Real^*_+$ be any decreasing function such that $T(N) \geq T(1)/N$ (Condition 1), cf.\  Figure~\ref{fig:ExemploReducao}(c). We will use $T(N)$ as an auxiliary function in order to build $T_{J_i}(N)$ for $1 \leq i \leq M$. Condition 1 and the fact that the function is decreasing guarantee that $T(N)$ is consistent with the two requirements that a job's running-time function must comply with, viz., that using more processors (up to the job's saturation point) makes the job run faster and that no speedup above $N$ is attainable. 
	
	In order to build each job's running-time function, let $\epsilon$ denote the smallest difference between any two distinct elements of the set $\set [ T(1),\ldots,T(S) ]$, that is, let $\epsilon = \min \set [ |T(N) - T(N-1)| \st 1 < N \leq S ]$. For $1 \leq N \leq S$, let $\prec_N$ be a linear order on $\set [1, \ldots, M]$ such that $c^{N}_i < c^{N}_{i'} \implies i \prec_N i'$. Let $T_{J_i}(N)$ be within the interval $(T(N), T(N) + \epsilon)$ such that $T_{J_i}(N) < T_{J_{i'}}(N) \iff i \prec_N i'$ (Condition 2). Let $T_{J_i}(w^j_i)$ be denoted by $t^j_i$.
	
	Let $\mathrm{CoI}(T, N)$ be such that: (i) $\mathrm{CoI}(t^j_i, w^j_i) = v^j_i$ for $1 \leq i \leq M, 1 \leq j \leq V_i$; (ii) if $i' \prec_N i''$, then for $T_{J_{i'}}(N) < t < T_{J_{i''}}(N)$, $\mathrm{CoI}(t, N) = c^{N}_{i''}$; (iii) for all points $(T, N)$ for which $\mathrm{CoI}$ remains undefined, let $\mathrm{CoI}(T, N) = 0$ if there exist $i, j$ such that $T \leq t^{j}_{i}$ and $N \leq w^j_i$, or let $\mathrm{CoI}(T, N) = \infty$, otherwise (Figure~\ref{fig:ExemploReducao}(d)). Note that, by Condition 2, $\mathrm{CoI}(T, N)$ is a function defined over all points. Moreover, the domination property, given in Section~\ref{intro}, holds by construction. 
	
	Now we choose $U_{J_i}(t)$ so that $\set[w^j_i \st 1 \leq j \leq V_i] = \mathcal{N}^*_{J_i}$ for all $1 \leq i \leq M$. We do so by letting $U_{J_i}(t) = \sum_{j = 1}^{V_i} \delta(t - t^j_i) C_{J_{i}}(w^{j}_i)$, where $\delta(x)$ is the Dirac delta function, that is, $\int_{x_1}^{x_2}\delta(y)\mathrm{d}y=1$ if and only if $x_1\leq 0$ and $x_2 \geq 0$ (the integral equals $0$, otherwise).
	
	By definition, AP returns a \textsc{Yes} answer on input $I$ if and only if there exists $N^\mathrm{opt}_{J_i} \in \mathcal{N}^*_{J_i}$ for every $1 \leq i \leq M$ so that $\sum_{1 \leq i \leq M} C_{J_i}(N^\mathrm{opt}_{J_i}) \geq K$ and $\sum_{1 \leq i \leq M} N^\mathrm{opt}_{J_i} \leq \mathit{MAXN}$. And since the members of $\mathcal{N}^*_{J_i}$ correspond to those of $\set[w^j_i \st 1 \leq j \leq V_i]$ (with the $0$ in the former set corresponding to the degenerate version in the latter), and moreover $C_{J_i}(w^j_i) = v^j_i$ and $\mathit{MAXN}=W$, readily we have that AP returns a \textsc{Yes} answer on input $I$ if and only if KVIP returns a \textsc{Yes} answer on input $I'$.
\end{proof}

\begin{figure}[!htbp]
\centering
\includegraphics[scale=0.37]{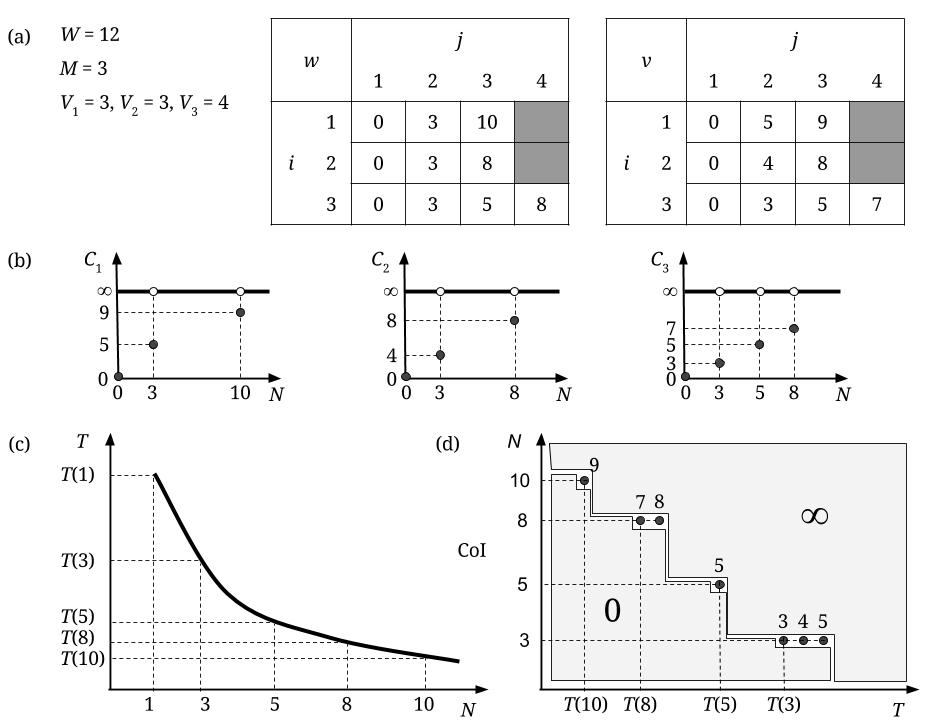}
\caption{Example of the reduction described in the proof of Theorem~\ref{APNPC}. Given an instance of KVIP (a), comprising the knapsack's capacity $W$, as well as the weight $w$ and value $v$ of each of the $M$ candidate items' versions, an instance of AP is constructed (b--d). This construction uses $\mathit{MAXN}=W$ and lets each item be a job. An item's versions are used to establish the corresponding job's utility function and thus its set $\mathcal{N}^*$ of optimal numbers of processors. For any given job $J_i$, this results in $N\in\mathcal{N}^*_{J_i}$ if and only if $N=w_i^j$ for some $1 \leq j \leq V_i$. Given an appropriate cost-of-infrastructure function (see below), the cost function that results for each job is as follows (b). The cost of executing job $J_i$ on $N$ processors, $C_i(N)=C_{J_i}(N)$, is finite only for $N\in\mathcal{N}^*_{J_i}$. When this is the case, we have $C_i(N)=C_i(w_i^j)=v_i^j$, where $j$ is the relevant version of item $i$. Given any decreasing function $T$ of $N$ such that $T(N) \geq T(1)/N$ (c), the cost-of-infrastructure function used to yield such job costs is the $\mathrm{CoI}(T,N)$ function depicted in (d). Specifically, we let $C_i(N)=\mathrm{CoI}(T_{J_i}(N),N)$ for each $N\in\mathcal{N}^*_{J_i}$, where $T_{J_i}(N)$ is derived from $T(N)$ in such a way as to yield different values for any two jobs and the same value of $N$. Take, for example, jobs $J_2$ and $J_3$, and note that $N=8$ is a member of both $\mathcal{N}^*_{J_2}$ and $\mathcal{N}^*_{J_3}$. In order to ensure the possibility of $C_2(8)\neq C_3(8)$, as in (b), it suffices that we require $T_{J_2}(8)\neq T_{J_3}(8)$, which in the example has been achieved by letting $T_{J_3}(8)=T(8)$ and $T_{J_2}(8)<T(8)+\epsilon$, with $\epsilon$ as in the theorem's proof. In general, jobs are considered in increasing cost order for fixed $N$ when effecting these deviations from $T(N)$.}
\label{fig:ExemploReducao}
\end{figure}

\section{Final remarks}
\label{conclusion}
	
	We have focused on the issue of determining the actual number of processors to be assigned to a program for distributed computation. This number is traditionally provided by users as an input parameter, even though as we have argued, its optimal determination can be rather involved. The current state of affairs just pushes the burden of such a decision towards the users.
    
	As an alternative, we have proposed that a server-side system should exist whose task would be to handle program allocation. Such a system would measure the load of the infrastructure and update the $\mathrm{CoI}$ function accordingly. For instance, the target could be to keep the infrastructure's load factor at some preestablished value (say, 90\% of all processors in the busy state, on average). Should the actual load factor fall below this threshold, costs would be decreased; they would be increased if the load factor grew above the threshold. The system should periodically solve the optimization version of KVIP, as described in Section~\ref{hardness}, in order to find the optimal solution to the optimization version of AP. All this would be based on known characterizations of the utility and running-time functions involved. Obtaining the latter, as we have noted, is not the goal of traditional analyses of algorithms, which express running-time functions in the big-oh notation and as such ignore constant factors, but rather should be approximated through some other form of algorithm analysis, one whose results were expressed in the more appropriate tilde notation.     

\subsection*{Acknowledgments}

The authors wish to thank CAPES, CNPq, FAPERJ, and a FAPERJ BBP grant.

\bibliography{mediation}
\bibliographystyle{unsrt}

\end{document}